\documentclass[12pt]{article}
\usepackage{setspace}
\usepackage{etex,mathpazo}
\usepackage{fullpage}
\usepackage{adjustbox}	
\usepackage{xr}
\usepackage{arydshln}
\usepackage{pdflscape}
\usepackage{amsmath,amsthm,amssymb} 
\usepackage{natbib} 
\usepackage[font=normalfont,tableposition=top]{caption}
\usepackage{float}
\usepackage{appendix}
\usepackage{color}
\usepackage{tikz}
\usetikzlibrary{arrows,intersections}
\usepackage{fullpage} 
\usetikzlibrary{arrows,intersections}

\usepackage{tikz}
\usetikzlibrary{patterns,hobby}
\usepackage{pgfplots}
\pgfplotsset{compat=1.6}
\usetikzlibrary{patterns,calc}
\usetikzlibrary{shapes,arrows}
\usepackage{hyperref}
\hypersetup{%
	colorlinks = true,
	linkcolor  = blue,
	citecolor  = blue
}

\makeatletter
\renewcommand*{\eqref}[1]{%
  \hyperref[{#1}]{\textup{\tagform@{\ref*{#1}}}}%
}
\makeatother

\usepackage{chngcntr}
\usepackage{setspace}

\newtheorem{theorem}{Theorem}[section]

\newtheorem{proposition}[theorem]{Proposition}

\newtheorem{remark}[theorem]{Remark}
\newtheorem{example}[theorem]{Example}
\newtheorem{examples}[theorem]{Examples}
\newtheorem{foo}[theorem]{Remarks}

\usepackage{bm}

\usepackage{graphicx}
\usepackage{adjustbox}	

\graphicspath{ {Figures/} }  	
\usepackage{graphics}
\usepackage{multirow}

\usepackage{caption}
\usepackage{subcaption}

\newcommand{\bbE}{\mathbb E}
\newcommand{\bbF}{\mathbb F}
\newcommand{\bbG}{\mathbb G}

\newcommand{\bbR}{\mathbb R}

\newcommand{\scF}{\mathcal F}

\newcommand\hasan[1]{{\leavevmode\color{black}{#1}}} 
\newcommand\kihun[1]{{\leavevmode\color{black}{#1}}}

\newcommand{\crl}[1]{\ensuremath{ \left\{ #1 \right\} }}

\newcommand{\levy}{{L\'evy~}}

  \clearpage 
\usepackage{titling}

\usepackage{booktabs}

\usepackage[left=1in, right=1in, top=1in, bottom=1in]{geometry}
\usepackage{setspace} 
\usepackage[T1]{fontenc}
\usepackage{caption}
\allowdisplaybreaks[4]

\usepackage{titlesec}
\titlespacing*{\section}{0pt}{0.5\baselineskip}{0.8\baselineskip}
\titlespacing*{\subsection}{0pt}{0.5\baselineskip}{0.01\baselineskip}

\usepackage{colortbl}

\expandafter\def\expandafter\normalsize\expandafter{%
    \normalsize
    \setlength\abovedisplayskip{3pt}
    \setlength\belowdisplayskip{3pt}
    \setlength\abovedisplayshortskip{50pt}
    \setlength\belowdisplayshortskip{50pt}
}

\usepackage{enumitem}

\usepackage{mathtools}
\mathtoolsset{showonlyrefs}

\usepackage[bottom]{footmisc}

\makeatletter
\def\@makefntext{\hskip 0em\@makefnmark}
\makeatother

\begin{document}


\title{\textbf{Time-changed \levy processes and option pricing: a critical comment}\thanks{We would like to thank Julien Hugonnier for motivating this note and for fruitful discussions. We also thank Samuel Cohen, Yan Dolinsky, Gregoire Loeper, Loriano Mancini, Stoyan Stoyanov, Farshid Vahid, and seminar participants at QMF~2018 and the 62nd Annual Meeting of the Australian Mathematical Society for comments. The Centre for Quantitative Finance and Investment Strategies has been supported by BNP Paribas.}}
\date{\today}
\author{
	Hasan Fallahgoul\thanks{Hasan A. Fallahgoul, Monash University, School of Mathematics and Centre for Quantitative Finance and Investment Strategies, 9 Rainforest Walk, 3800 Victoria, Australia. E-mail: hasan.fallahgoul@monash.edu.}\\{\small Monash University}
	\\
	\bigskip
	\and Kihun Nam\thanks{Kihun Nam, Monash University, School of Mathematics and Centre for Quantitative Finance and Investment Strategies, 9 Rainforest Walk, 3800 Victoria, Australia. E-mail: kihun.nam@monash.edu.}\\{\small Monash University}
	\\
	\bigskip}
\noindent
\maketitle

\begin{abstract}
	\cite{Carr:2004hl}, henceforth CW, developed a framework that encompasses almost all of the continuous-time models proposed in the option pricing literature. Their framework hinges on the stopping time property of the time changes. By analyzing the measurability of the time changes with respect to the underlying filtration, we show that all models CW proposed for the time changes fail to satisfy this assumption.
	
	%
	
\end{abstract}

\setcounter{page}{1}


\setstretch{2}





\cite{Carr:2004hl}, henceforth CW, developed a time-changed L\'evy (TCL) framework that encompasses almost all of the continuous-time models proposed in the option pricing literature by allowing a correlation between the L\'evy process and stochastic time. The whole paper is based on their closed-form expression of the characteristic function for the TCL process,\footnote{The closed-form expression for the characteristic function in CW is given by the Laplace transform of stochastic time $ T $ evaluated at the characteristic exponent of $ X $ under the \textit{leverage-neutral measure}, which is a complex measure.} which hinges on the stopping time property for the stochastic time.\footnote{For a given filtration $(\scF_s)_{s\geq 0}$, we say a stochastic time $T=(T_t)_{t\geq 0}$ has the stopping time property if $\crl{T_t\leq s}\in\scF_s$  for all $s,t\geq 0$.} Though their closed-form expression for the characteristic function is correct when the stochastic time is a collection of stopping times, none of the specifications proposed by CW for the time changes satisfies this condition, as the stochastic time is given only as a process adapted to the underlying filtration. 

It is not obvious, if not impossible, to construct the stochastic time $T=(T_t)_{t\geq 0}$ such that (i) it is correlated to the L\'evy process; (ii) the time change $T_t$ is a stopping time for each $t$; (iii) the Laplace transforms of $T_t$s with respect to the \textit{leverage-neutral complex measure} can be expressed in closed form; and (iv) it captures empirical regularities such as stochastic volatility and volatility clustering. In fact, it is more natural to assume that the stochastic time is an adapted process. In the CW's framework, the asset log-return is modeled by $X_T$, a L\'evy process $X$ that runs on a continuous stochastic time $T$. By the L\'evy-It\^o decomposition theorem, the continuous part of the L\'evy process $X$ is a Brownian motion with drift. Since $T_t=\frac{t}{\bbE Q_t}Q_t$ where $Q$ is the continuous part of the quadratic variation of the asset log-return $X_T$, we know $T$ can be observed from the market. Therefore, the adaptedness assumption on $T$, which was used in the specifications in CW and literature, is appropriate in the financial context. 

One possibility to apply the CW's framework to their specifications is to change the filtration so that the $T_t$s become stopping times and $X$ remains L\'evy. If the L\'evy process $X$ and the stochastic time $T$ are independent, then the filtration change poses no problem. However, when $X$ and $T$ are dependent, the filtration change that makes $T_t$s stopping times will affect the semimartingale property of $X$, and this will lead us to an arbitrage opportunity. Let us clarify with an example. For a positive constant $\rho$ and two independent Brownian motions $W$ and $B$, consider the case where $X=\rho W+\sqrt{1-\rho^2}B$ and $T_t=\int_0^t\exp(W_s-s^2/2)ds$. The minimal filtration that makes $T_t$s stopping times and $X$ adapted is $\bbG:=(\scF^W_{C_t}\vee \scF^X_t)_{t\geq 0}$, where $C_t$ is the first $t$-level crossing time of $T$ (see Proposition \ref{pre_stopping_adapted0}). This implies that under $\bbG$, one can determine the value of $C_t$ and the path of $W$ from time $0$ to $C_t$ at time $t$. If $C_t>t$ happens at $t$, $X$ from time $t$ to $C_t$ can be decomposed into $\bbG$-martingale $\sqrt{1-\rho^2}B$ and \textit{deterministic} path $\rho W$ with infinite first-order variation. Therefore, $X$ is no longer a semimartingale under $\bbG$ and it has an arbitrage opportunity. In particular, if we set $\rho=1$, then we can set an arbitrage strategy as follows: Since $C_t$ is adapted to $\bbG$, we can determine whether $C_t>t+1$ or not at time $t$. At time $t$, if $C_t>t+1$, then we can find the path of $X$ from $t$ to $C_t$ and the asset price at time $t+1$. If the asset price at $t+1$ is higher than the asset price at $t$, you buy the asset at time $t$, otherwise, you short the asset.

The above argument is based on Proposition \ref{pre_stopping_adapted0} which tells us the following: Under the enlarged filtration that makes $T$ being stopping time and adapted, we are able to foretell, not forecast, the future of the business activity rate $v=\frac{dT}{dt}$ as well as the past of it. Let us provide an example from CW. Let the business activity rate $ v $ be a Cox-Ingersoll-Ross (CIR) process, which is an affine activity rate model (see Section 4.2.1 of CW). The left panel of Figure \ref{simu_cir} provides a simulation of the CIR process with parameters from Table 1 in \cite{fallahgoul2019time}, while the right panel shows the integral of the simulated business activity rate that represents $T$. Since the calendar time $ t $ is replaced by the time change $ T_t $, one needs to impose the condition that the unconditional expectation of $ T_t $ is $ t $. This implies that a realized path of $T$ always oscillates around the line $y=t$, as the right panel of Figure \ref{simu_cir} shows. Since the business activity rate $v$ is adapted, its integrated process $T$ is adapted as well. If $ T_t $ is a stopping time for each $t$, then based on Proposition 1, we know the path of $ T $ until it hits the level of $ t $ at time $t$. Assume $ T_t >t$ at time $ t=1.2 $ as in the right panel of Figure 1. Then, at time $t=1.2$, we know the whole path until $ S $ using the adaptedness of $T$. Since the stopping time property of $T_t$ for all $t$ requires knowledge of the path only until $R$, there is no problem in this respect. Alternatively, assume that $ T_t < t$ at time $ t=0.6$, as in the right panel of Figure 1. Then, the stopping time property of $ T_t $ for each $t$ implies that we know the whole path of $T$ until $Q$ instead of $P$ at time $t=0.6$. This is impossible under the natural filtration that defines the CIR process. Therefore, we cannot use the CIR process for model stopping times in CW's framework.

\begin{center}
	\fbox{Insert Figure~\ref{simu_cir} about here}
\end{center}


In sum, unfortunately, there are several ways to show that none of the specifications proposed in CW for the stochastic time satisfies their assumptions. Perhaps more importantly, the construction of stochastic processes that satisfy their assumptions and are empirically friendly, meaning one can obtain their characteristic functions from CW's results, is challenging and still an open topic for further investigation.

\cite{fallahgoul2019corr} extend CW's result to a case in which the stochastic time $T$ is just an adapted process. The model becomes flexible enough to incorporate virtually all models proposed in the option pricing literature. In this case, all specified option pricing models in \cite{Huang:in} are nested in our setting.

\let\clearpage\relax

\renewcommand{\baselinestretch}{1} 
\normalsize
\bibliographystyle{rfs}
\bibliography{untitled}
\appendix
\setstretch{1.5} 

\begin{figure}[!ht]
	\renewcommand{\thefigure}{\arabic{figure}}
	\begin{center}
		\includegraphics[width=0.8\textwidth]{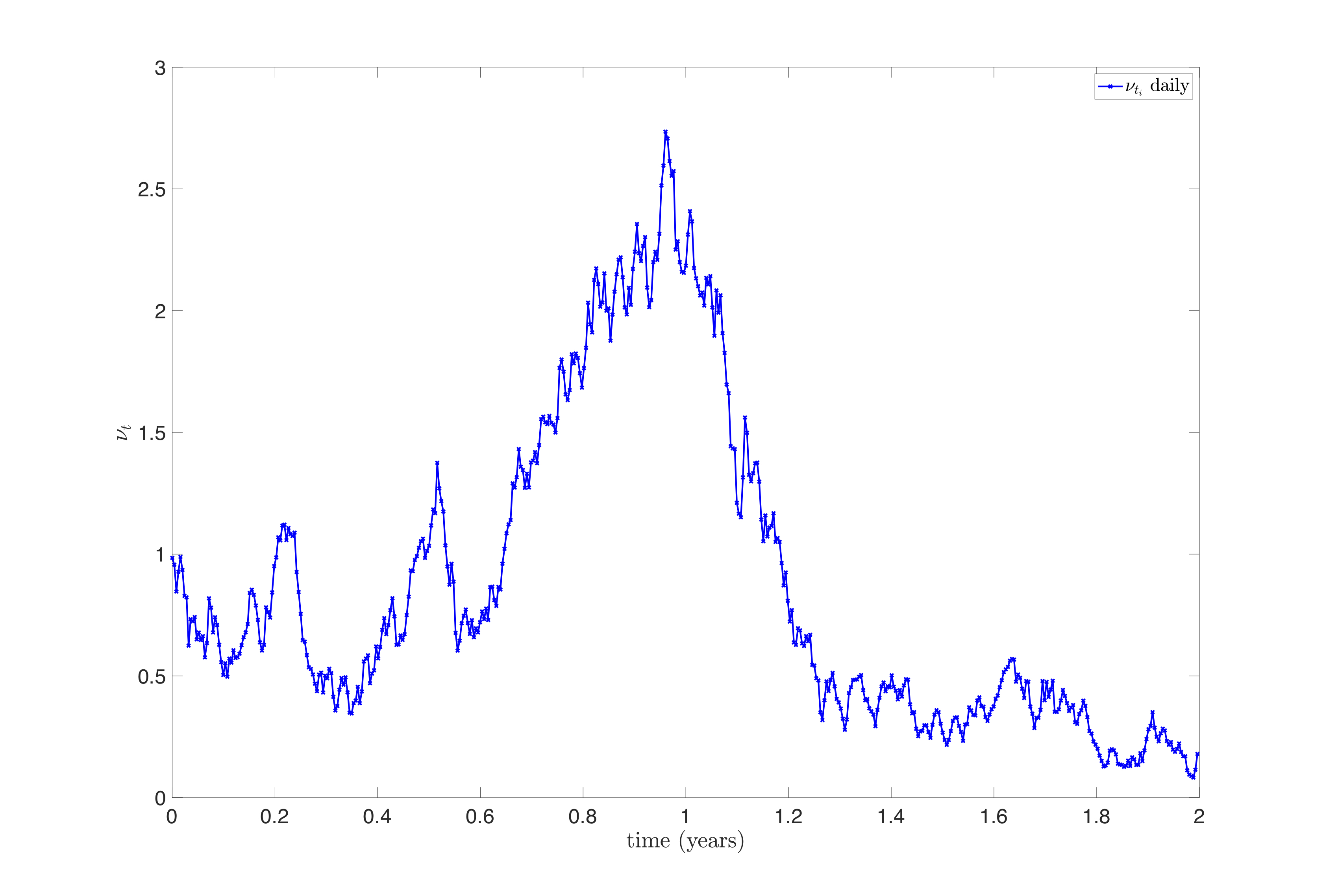}  
		\includegraphics[width=0.8\textwidth]{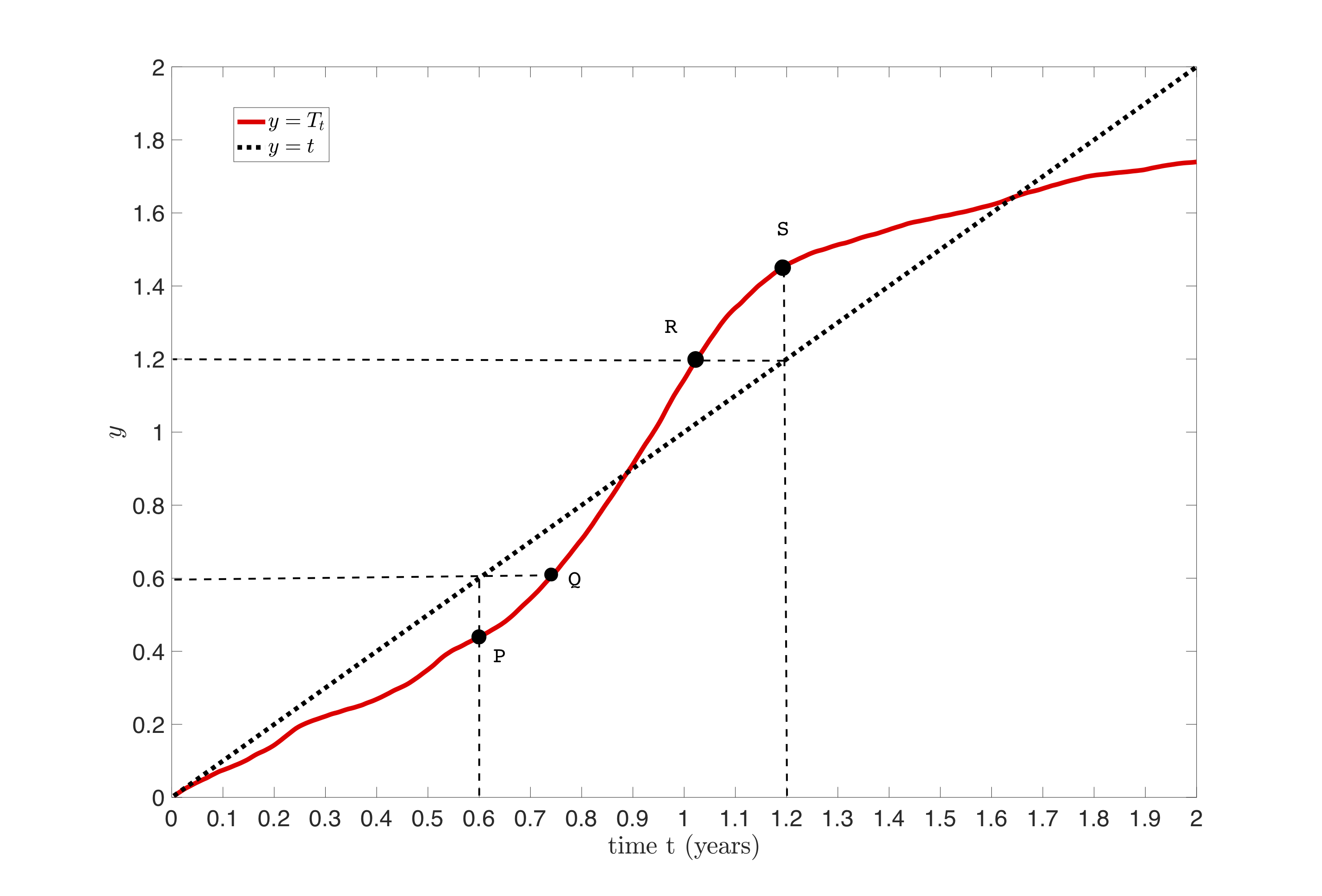}
	\end{center}     
		\caption{Simulated CIR process.}              
	\caption*{(top) Two years simulated Cox-Ingersoll-Ross (CIR) process, i.e., $ dv_t = \kappa(\theta-v_t)\, dt + \sigma\sqrt{v_t}\, dW_t $,  where $ W_t $ is a Brownian motion, parameter $ \kappa $ corresponds to the rate of mean reversion, $ \theta $ is the long-runs mean, and  $ \sigma $ capturs volatility. (down) Integrated CIR process. $ T_t=\int_{0}^{t} v_s ds$.}\label{simu_cir}
\end{figure} 
\newpage
\section{Adaptedness and the stopping time property}
Let $\bbF=(\scF_t)_{t\geq 0}$ be the underlying filtration in CW.\footnote{Throughout CW, the filtration they refer to when they mention martingale, stopping time, and the \levy process is not clear. To make Lemma 1 of CW correct, we need to assume that $T_t$ are stopping times with respect to filtration $\bbF$ such that $X$ is the \levy process, while CW assume that $T$ is adapted to $\bbF$ in Section 4.3.1. For consistency, we assume that CW used the same filtration throughout their paper.} \kihun{In their paper, the stochastic time $T$ is an absolutely continuous process and satisfies two assumptions: $T_t$ is a $\bbF$-stopping time for each $t$, and for $s\leq t$, $T_s$ is $\scF_t$-measurable.}
\hasan{The first assumption \kihun{is used to derive} the characteristic function or generalized Fourier transform of the TCL process (see Lemma 1 of CW), while the \kihun{second assumption is used} to specify the business activity rate $v_t=\frac{dT_t}{dt}$ in Section 4.2 of CW.} 

These two assumptions imply that, at time $t$, we can determine when the stochastic time $T$ hits the $t$-level, as well as the path of $T$ until it hits the $t$-level. 

\begin{proposition}\label{pre_stopping_adapted0}
	Let $(T_t)_{t\geq 0}$ be an increasing continuous process. Assume that $\bbF:=(\scF_t)_{t\geq 0}$ is the filtration such that $T$ is adapted and $T_t$ is a $\bbF$-stopping time for each $t\geq 0$. Let $C_s:=\inf\crl{t> 0: T_t>s}$ be the first $s$-level crossing time of $T$. Then, for all $t\geq 0$, $C_t$ is $\scF^T_t$-measurable. Moreover,
	\begin{align*}
	\scF^T_{C_t}\subset \scF^T_t\subset\scF_t
	\end{align*}
	where $(\scF^T_t)_{t\geq 0}$ is the filtration generated by $T$.
\end{proposition}
\begin{proof}
	Note that $C_t$ is $\scF^T_t$-measurable for every $t\geq 0$ because
	\begin{align*}
	C_t^{-1}([s,\infty))=\crl{C_t\geq s}=\crl{T_s\leq t}\in\scF^T_t
	\end{align*}
	for any $s\in\bbR$ and $\crl{[s,\infty):s\in\bbR}$ generates Borel sigma algebra in $\bbR$. Therefore, $\scF^C_t\subset\scF^T_t$, where $(\scF^C_t)_{t\geq 0}$ is the filtration generated by $C$. We now prove $\scF^T_{C_t}\subset \scF^C_t$. 
	
	Since $T$ is a continuous increasing function,
	$
	T_s=\inf\crl{u\geq 0:C_u>s}.
	$
	Therefore, $T_s$ is the first $s$-level crossing time of $C$. This implies that
	$
	\scF^T_s\subset\scF^C_{T_s}.
	$
	Now, if we let $s=C_t$ and use the relationship $T_{C_t}=t$, we have
	\begin{align*}
	\scF^T_{C_t}\subset \scF^C_t
	\end{align*}
	which completes the proof.
\end{proof}

\end{document}